\documentclass[11pt,a4paper]{article}
\usepackage[margin=1in]{geometry}
\usepackage[T1]{fontenc}
\usepackage{lmodern}
\usepackage{microtype}
\usepackage{mathtools,amssymb,amsthm}
\usepackage[numbers,sort&compress]{natbib}
\usepackage{tikz}
\usetikzlibrary{arrows.meta,positioning,calc,decorations.pathreplacing}
\usepackage[hidelinks]{hyperref}

\theoremstyle{plain}
\newtheorem{theorem}{Theorem}[section]
\newtheorem{lemma}[theorem]{Lemma}

\theoremstyle{definition}
\newtheorem{definition}[theorem]{Definition}
\newtheorem{remark}[theorem]{Remark}

\newcommand{\ppE}{\mathbb E}
\newcommand{\ppR}{\mathbb R}

\newcommand{\ppFill}[1]{\mathsf A_{#1}}
\newcommand{\ppSwap}[1]{\mathsf T_{#1}}
\newcommand{\ppId}{\mathrm I}

\title{Asymptotically Optimal Public Project Redistribution without Bounded Precision: A Complex-Analytic Proof}
\author{Mingyu Guo\\[0.4em]
  \small School of Computer Science and Information Technology\\
  \small Adelaide University\\
  \small \href{mailto:mingyu.guo@adelaide.edu.au}{\texttt{mingyu.guo@adelaide.edu.au}}}
\date{}
\hypersetup{
  pdftitle={Asymptotically Optimal Public Project Redistribution without Bounded Precision: A Complex-Analytic Proof},
  pdfauthor={Mingyu Guo}}

\begin{document}
\maketitle

\begin{abstract}
We study worst-case VCG redistribution mechanism design for the
public project problem, where $n$ agents decide whether to build a
non-excludable public project and the designer chooses a Groves term
that returns as much of the VCG payment as possible without running
a deficit. The objective is the worst-case efficiency ratio, the
worst-case ratio between the agents' total utility and the first-best
total utility. Prior work showed that this ratio can approach $1$ as
$n$ grows only under a bounded precision assumption, that all types
are rational numbers with a common bounded denominator. The
assumption is stronger than it looks: it confines the difficult
profiles to those in which all but a bounded number of agents report
zero, and an agent who reports zero knows the total exactly. Without
it, the best known guarantee tends to $1/2$. We remove the
assumption. We
construct a deterministic, anonymous, strategy-proof, and
non-deficit mechanism for arbitrary real types in $[0,1]$ whose
worst-case efficiency ratio is $1-O(1/\log\log n)$. The construction
estimates the shortfall of the reported total below the project cost
by filling in each agent's missing report with a resampled one, then
estimates the error of that estimate, then the error of the error,
and so on. After $m$ rounds the aggregate error is at most
$1/H_m$ mean reports, where $H_m$ is the $m$th harmonic number. The
proof of this bound is the heart of the paper and, unusually for a
result about payments, uses complex analysis: a trigonometric
representation of the shortfall function together with the
Poisson--Jensen inequality for bounded analytic functions in the
upper half plane.
\end{abstract}


\section{Introduction}\label{pp:sec:intro}

The public project problem~\cite{pp-moulin1988}
is a fundamental mechanism design model that has been well studied
in both computer science and economics. A group of agents decide
whether or not to build a non-excludable public project, such as a
bridge that can be accessed by everyone. The mechanism decision is
binary, build or not build, and once built the project can be
consumed by every agent regardless of her contribution. In this
setting the Vickrey--Clarke--Groves (VCG) mechanism is the canonical
solution, as it is efficient and dominant-strategy incentive
compatible, but it suffers from a major drawback: budget surplus.
Surplus payments leave the system and are wasted. VCG
redistribution mechanisms were introduced to recover this surplus
by returning as much of the collected VCG payments to the agents as
possible, without affecting incentives and without ever running a
deficit. Our model follows prior work on VCG redistribution
mechanism design for the public project
problem~\cite{pp-naroditskiy2012,pp-guo2016,pp-guoshen2017,pp-guo2019,pp-wang2021wiiat,pp-guo2023lottery}.
The objective, following that literature, is the worst-case
efficiency ratio: the fraction of the first-best welfare that the
mechanism guarantees at every type profile. The question we settle
is whether this ratio can approach $1$ as the number of agents
grows, a question raised by Naroditskiy et
al.~\cite{pp-naroditskiy2012} and open since 2012. We prove that it can, for arbitrary
real types, with a deterministic mechanism. The paper is, in
effect, one construction and one proof; perhaps surprisingly for a
paper about payments, the proof turns on complex analysis.

\subsection{Problem description}\label{pp:sec:model}

\begin{definition}[Public project problem]
There are $n\geq3$ agents who decide whether or not to build a
public project that costs $1$. Agent $i$'s type is
$\theta_i\in[0,1]$. If the decision is to build, agent $i$'s
valuation is $\theta_i$. If the decision is not to build, agent
$i$ retains her share $1/n$ of the project cost, and her valuation
is $1/n$.
\end{definition}

It is without loss of generality to assume that valuations are
bounded above by the project cost~\cite{pp-naroditskiy2012}.
Each agent reports a type in $[0,1]$. For a report profile
$\theta=(\theta_1,\ldots,\theta_n)$ write $s=\sum_i\theta_i$ for
the reported total, $\theta_{-i}$ for the reports of the agents
other than $i$, and $s_{-i}=\sum_{j\ne i}\theta_j$ for their sum.
VCG redistribution mechanisms all share the same allocation rule:
the project is built if and only if the agents' total valuation
reaches the project cost, that is, if and only if $s\geq1$. The
total valuation under this efficient decision is
\[
 S(\theta)=\max\{s,1\}=s+g,
 \qquad\text{where}\quad g=G(s),\quad G(t)=\max\{1-t,0\}.
\]
We call $G$ the \emph{shortfall function} and $g$ the
\emph{shortfall}: the amount by which the reported total falls
short of the project cost, and zero when it does not. The
shortfall plays a central role below. We use this case convention
throughout: an uppercase letter is a
function of a real argument, and the matching lowercase letter is
its value at the total $s$.

VCG redistribution mechanisms first charge the VCG payments, which
ensures strategy-proofness, and then return to every agent a
redistribution payment that is independent of her own type, which
ensures that the redistribution affects neither efficiency nor
strategy-proofness. By Holmstr\"om's
characterization~\cite{pp-holmstrom1979}, for our model these are
the only mechanisms that are efficient and strategy-proof; they are
the Groves mechanisms~\cite{pp-groves1973}. We restrict attention
to anonymous mechanisms, in which all agents are treated by the
same rule; for the worst-case efficiency objective adopted below,
and throughout this literature, this is without loss of
generality~\cite{pp-apt2008,pp-naroditskiy2012}.

\begin{definition}[VCG redistribution mechanisms for the public
project problem~\cite{pp-naroditskiy2012}]
Build the project if and only if $s\geq1$. Agent $i$ is charged
\begin{equation}\label{pp:eq:payments}
 p_i(\theta)=
 \begin{cases}
  h(\theta_{-i})-s_{-i},&\text{if built},\\
  h(\theta_{-i})-(n-1)/n,&\text{if not built},
 \end{cases}
\end{equation}
where $h:[0,1]^{n-1}\to\ppR$ is applied to the other agents'
reports. A positive $p_i$ is a charge and a negative $p_i$ is a
transfer to the agent.
\end{definition}

The function $h$, often called the Groves term, characterises a
specific VCG redistribution mechanism, and the mechanism design
task comes down to designing $h$. Because $h(\theta_{-i})$ does not
depend on agent $i$'s report, her report affects her utility
$v_i-p_i$ only through the decision, and the decision maximises her
true value plus the others' reported values exactly when she
reports truthfully; so truth-telling is a dominant strategy for
every choice of $h$.

There is only one mechanism design constraint. Let
$L(\theta)=\sum_ip_i(\theta)$ be the revenue retained by the
mechanism. Summing~\eqref{pp:eq:payments} at truthful reports gives
\begin{equation}\label{pp:eq:welfare}
 L(\theta)=\sum_ih(\theta_{-i})-(n-1)S(\theta).
\end{equation}
The mechanism is \emph{ex post weakly budget balanced}, or
\emph{non-deficit}, if $L(\theta)\geq0$ at every profile. The agents'
total utility is $W(\theta)=S(\theta)-L(\theta)$: retained revenue
is a welfare loss. The first-best total utility, the maximum total
utility assuming the agents do not lie, is $S(\theta)$, attained
when the efficient decision is made and no money leaves the
system. Following Moulin~\cite{pp-moulin2009}, we measure a
mechanism $M$ by its \emph{worst-case efficiency ratio}, the
worst-case ratio between the achieved total utility and the
first-best total utility:
\[
 \alpha(M)=\inf_{\theta\in[0,1]^n}\frac{W(\theta)}{S(\theta)}
         =1-\sup_{\theta\in[0,1]^n}\frac{L(\theta)}{S(\theta)}.
\]
The design problem is to choose $h$ so as to maximise $\alpha(M)$.
The constraint and the objective can be summarised in one
inequality, as in~\cite{pp-guo2023lottery}:
\begin{equation}\label{pp:eq:one-inequality}
 (n-1)S(\theta)\leq\sum_ih(\theta_{-i})\leq(n-\alpha(M))S(\theta)
 \qquad\text{for all }\theta .
\end{equation}
A weakly budget-balanced mechanism has $\alpha(M)\leq1$, and
$\alpha(M)$ is close to $1$ exactly when $\sum_ih(\theta_{-i})$ is
close to $(n-1)S(\theta)$ at every profile. As in prior work on
this model~\cite{pp-naroditskiy2012,pp-guo2019,pp-guo2023lottery},
non-deficit is the only constraint: we impose neither individual
rationality nor nonnegativity of individual transfers.

\subsection{Main result}\label{pp:sec:mainresult}

Naroditskiy et al.~\cite{pp-naroditskiy2012} conjectured that
their upper bounds on the worst-case efficiency ratio are tight;
since the bounds approach $1$, this would mean that the ratio can
approach $1$ as $n\to\infty$. Guo~\cite{pp-guo2019} proved the
asymptotic statement under the following assumption.

\begin{definition}[Bounded precision assumption~\cite{pp-guo2019}]
Every agent's type is a rational number whose denominator is
bounded above by a constant $C$.
\end{definition}

\noindent Without the assumption, the best guarantee known, also
from~\cite{pp-guo2019}, tends to $1/2$.

Guo~\cite{pp-guo2019} argues that the assumption is mild, since
currency values and fixed-precision formats satisfy it. It does
more than it appears to, however. Its proof first shows that the
mechanism loses nothing at profiles with total at least $4$, so
only profiles with total below $4$ matter; there a positive type is
at least $1/C$ and so at most $4C$ agents report a positive type,
and all other agents report zero. An agent who reports zero knows the total $s$ exactly
from her own partial total $s_{-i}$, and the mechanism's estimate is
exact whenever at least two of three sampled agents report zero,
which under the assumption happens for all but an $O(C^2/n)$
fraction of samples. In other words, the assumption reduces the
problem to profiles in which almost every agent can compute the
shortfall exactly by assuming her own type is zero, and asymptotic
optimality follows from counting the exceptions. Dropping the
assumption removes this crutch entirely: at the profile in which
every agent reports $1/n$, the mechanism of~\cite{pp-guo2019}
retains $3(n-1)/n$ times the first-best welfare, so its worst-case
efficiency ratio over arbitrary real types is $-2+3/n<0$. The
substantive guarantee of that paper for arbitrary types is
therefore its second mechanism, with ratio $(n+1)/(2n)$, which
tends to $1/2$. That is far below the conjectured optimal ratios,
which increase with $n$ and already equal $2/3$ at
$n=3$~\cite{pp-naroditskiy2012}. The present paper removes the
assumption and replaces counting by averaging: every bound is in
units of the mean report $s/n$, so many tiny positive reports are
no harder than a few large ones.

\begin{theorem}[Asymptotic optimality without bounded precision]
\label{pp:thm:asymptotic}
There are mechanisms $M_n$, one for each $n\geq3$, each
deterministic, anonymous, efficient, dominant-strategy incentive
compatible, and ex post weakly budget balanced on arbitrary real
types in $[0,1]^n$, such that
\[
 0\leq1-\alpha(M_n)=O\!\left(\frac1{\log\log n}\right)
 \qquad(n\to\infty).
\]
\end{theorem}

The guarantee is worst-case: it holds at every profile and
makes no distributional assumption. Here and below, $\log$
denotes the natural logarithm. The rest of the paper constructs
$M_n$ and proves the theorem.

\section{Related Research}\label{pp:sec:related}

\paragraph{VCG redistribution.}
VCG redistribution mechanisms have been studied extensively in
resource allocation settings. Cavallo~\cite{pp-cavallo2006}
proposed a redistribution mechanism applicable to any VCG setting:
each agent receives $1/n$ of the smallest total VCG payment
consistent with the others' reports. For single-item auctions it is
asymptotically optimal, with worst-case efficiency ratio approaching
$1$ as the number of agents grows, and Guo~\cite{pp-guo2011aaai}
obtained a similar result for combinatorial auctions with gross
substitutes. For fixed numbers of agents, worst-case optimal
mechanisms are known for multi-unit and heterogeneous-item
auctions~\cite{pp-moulin2009,pp-guoconitzer2009,pp-gujar2011,pp-guo2012aamas}.
In these settings the optimal redistribution is linear in the order
statistics of the bids, and the designs for richer settings
generalise this one-dimensional linear structure. The public
project setting is different: the efficient decision depends on
the total of all reports, the optimal redistribution is not linear
even for three agents~\cite{pp-naroditskiy2012}, and manual
derivation of optimal mechanisms has proved far harder.

\paragraph{The public project problem, 2012 to 2019.}
Naroditskiy et al.~\cite{pp-naroditskiy2012} initiated worst-case
VCG redistribution mechanism design for the public project problem.
They proved an upper bound on the achievable worst-case efficiency
ratio for each $n$, conjectured it to be tight, derived the optimal
mechanism for $n=3$, and proposed the sampling-based redistribution
(SBR) heuristic: each agent estimates the total VCG payment by
assuming that her own type is drawn uniformly from the other
reports, after which a Cavallo-style charge of $1/n$ of the largest
possible deficit restores the non-deficit constraint. The fill-in
operation of the present paper is this heuristic applied to the
shortfall. The upper bounds follow a pattern that approaches $1$ as
$n\to\infty$, which is the origin of the conjecture that the
asymptotic worst-case ratio is $1$.
Subsequent work for $n\geq4$ was heuristic or computational. The
ABR mechanism~\cite{pp-guo2016} achieves a constant asymptotic
worst-case ratio of $0.102$. Automated mechanism design with sampled
worst-case profiles~\cite{pp-guoshen2017} gives better ratios for
small $n$ but is computationally expensive to run, and the ratios
of the SBR and ABR heuristics had to be computed numerically; as a
result, ratios for these mechanisms were known only for at most
$10$ agents~\cite{pp-guo2019}.
Guo~\cite{pp-guo2019} then gave a mechanism with an analytic
worst-case ratio of $(n+1)/(2n)$, hence $1/2$ asymptotically, for
arbitrary types, and a second mechanism whose worst-case ratio tends
to $1$ under the bounded precision assumption discussed in
Section~\ref{pp:sec:mainresult}.

\paragraph{Neural and automated approaches, 2018 to the present.}
In response to the analytic barrier, recent work turns to neural
networks: Manisha et al.~\cite{pp-manisha2018} learn redistribution
mechanisms for auctions, Wang et al.~\cite{pp-wang2021aamas} study
neural mechanism design for public projects, though not in the
redistribution setting, and Wang et al.~\cite{pp-wang2021wiiat}
learn redistribution for the public project problem with worst
cases estimated by Monte Carlo sampling.
Guo~\cite{pp-guo2023lottery} combined neural training with exact
worst-case verification by mixed integer programming, which
revealed that sampled worst cases are substantially inflated, and
found the optimal mechanisms for up to $5$ agents, confirming the
tightness of the upper bounds conjectured in~\cite{pp-naroditskiy2012}
for those $n$.
These methods produce certified mechanisms for small fixed $n$. They
say nothing about the limit $n\to\infty$, which is the subject of
this paper.

\paragraph{Contribution.}
We give the first deterministic mechanism for the public project
problem that is efficient, strategy-proof, and non-deficit for
arbitrary real types in $[0,1]$ and whose worst-case efficiency
ratio tends to $1$, at rate $1-O(1/\log\log n)$, settling the
asymptotic question raised by~\cite{pp-naroditskiy2012}. The mechanism is a finite
averaging scheme that estimates the shortfall, then the error of
the estimate, then the error of the error. Its analysis is where
the novelty lies: the rate rests on a bound for an iterated
averaging operator that we prove through a trigonometric
representation of the shortfall function and the Poisson--Jensen
inequality for bounded analytic functions in the upper half plane,
applied to the characteristic function of a uniformly drawn report.
Complex analysis is rare in mechanism design, where arguments are
typically combinatorial, probabilistic, or linear-programming
based; to our knowledge this is its first use in the design of VCG
redistribution payments.

\section{The Idea}\label{pp:sec:idea}

\begin{figure}[t]
\centering
\begin{tikzpicture}[
  font=\small,
  box/.style={draw,rounded corners=2pt,inner sep=4pt,minimum width=3.4cm,align=center},
  note/.style={font=\footnotesize,text=black!70,align=left,anchor=west},
  arr/.style={-{Stealth[length=5pt]},thick},
  lab/.style={font=\footnotesize,right,midway,xshift=1pt},
  node distance=7mm]
 \node[box] (G)  {target $g=G(s)$};
 \node[box,below=of G]  (G1) {$G_1=\mathsf TG$};
 \node[box,below=of G1] (G2) {$G_2=G_1+\mathsf TE_1$};
 \node[below=4mm of G2,inner sep=1pt] (dots) {$\vdots$};
 \node[box,below=4mm of dots] (Gm) {$G_m$,\quad $|\varepsilon_m|\le s/(nH_m)$};
 \node[box,below=of Gm,fill=black!5] (phi) {$\phi$: same recipe with $\mathsf A_{-i},\mathsf T_{-i}$};
 \node[box,below=of phi,fill=black!5] (h) {$h=s_{-i}+\tfrac{n-1}{n}\phi+\tfrac{b_{n,m}}{n}s_{-i}$};
 \draw[arr] (G)  -- node[lab] {fill in $G$} (G1);
 \draw[arr] (G1) -- node[lab] {fill in $E_1$} (G2);
 \draw[arr] (G2) -- (dots);
 \draw[arr] (dots) -- node[lab] {fill in $E_{m-1}$} (Gm);
 \draw[arr] (Gm) -- node[lab] {replace $\mathsf A,\mathsf T$: cost $\le4m2^ms/n^2$} (phi);
 \draw[arr] (phi) -- node[lab] {surcharge} (h);
 \node[note,left=4mm of G1.west,anchor=east,align=right] {error\\$E_1=(\mathrm I-\mathsf T)G$};
 \node[note,left=4mm of G2.west,anchor=east,align=right] {error\\$E_2=(\mathrm I-\mathsf T)^2G$};
 \node[note,left=4mm of Gm.west,anchor=east,align=right] {error\\$E_m=(\mathrm I-\mathsf T)^mG$};
 \node[note,left=4mm of h.west,anchor=east,align=right] {$0\le L\le2b_{n,m}S$\\$b_{n,m}=\tfrac1{H_m}+\tfrac{4m2^m}{n}$};
\end{tikzpicture}
\caption{The construction at a glance. Above the shaded boxes:
the analysis with the full-list operations, each round filling in
the current error function; after $m$ rounds the aggregated
estimate is $g_m=G_m(s)$ with error $\varepsilon_m$, bounded by
Lemma~\ref{pp:lem:approximation}. Shaded boxes: the mechanism
replaces $\mathsf A,\mathsf T$ by their admissible versions at a
cost bounded by Lemma~\ref{pp:lem:deletion}, and a surcharge
calibrated to the total error bound makes it non-deficit
(Theorem~\ref{pp:thm:main}).}
\label{pp:fig:idea}
\end{figure}
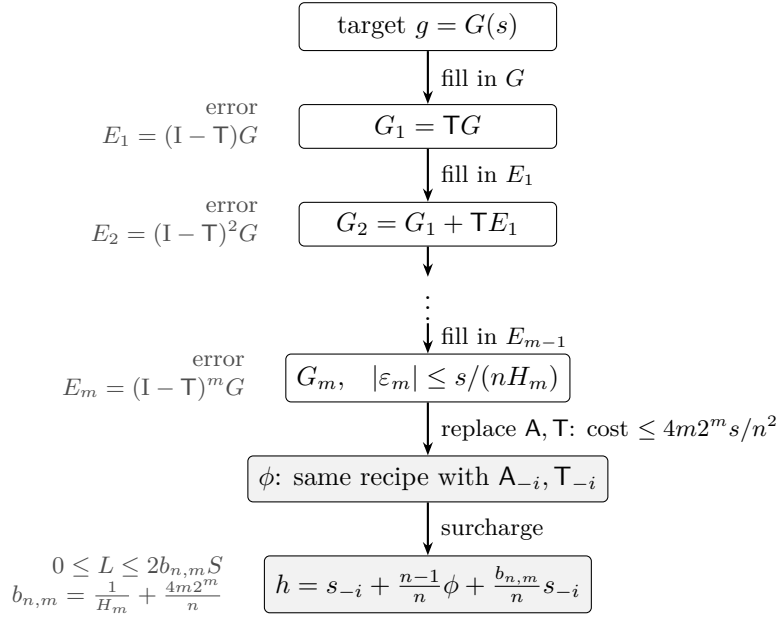

Figure~\ref{pp:fig:idea} summarises the construction described in
this section.

Consider the Groves term
\[
 h(\theta_{-i})=s_{-i}+\frac{n-1}{n}\,g.
\]
Summed over the agents it gives exactly $(n-1)S(\theta)$, so the
mechanism would collect no revenue and achieve the first-best
welfare. It is not admissible, because the shortfall $g=G(s)$
depends on the total $s$, which includes agent $i$'s own report.
Throughout the paper we are therefore interested in evaluating
functions of the total, such as $G(s)$, when agent $i$ has access
only to the other reports $\theta_{-i}$, and hence only to her
partial total $s_{-i}$.

\paragraph{Filling in.}
A natural trick, used by Naroditskiy et al.~\cite{pp-naroditskiy2012}
to estimate the total VCG payment, is to draw a report uniformly
from $\theta_{-i}$ and use it in place of the missing one, averaging
over the draw. Unfortunately this is awkward to work with
mathematically, because
every agent draws from a different list. A much cleaner structure
is for every agent to draw uniformly from the full profile
$\theta$, her own report included. This is not admissible: sometimes
she draws her own report, which is a value she must not use. We call
it the \emph{fill-in} operation all the same, and we allow it
throughout the analysis, because with it the error has an exact
form at every stage. At the end we show that the inadmissible draw
can be replaced by the admissible one, drawing from the other
reports only, and that the swap changes the outcome by a negligible
amount as $n\to\infty$.

\paragraph{Error, then error of the error.}
We use the fill-in to estimate $G(s)$. The estimate has an error,
which we call the first-order error. The point that makes the
construction work is that this error is again a function of $s$:
not a number anyone knows, but a known function evaluated at the
unknown total. So we can estimate it by the same fill-in and add
the estimate back. That does not eliminate the first-order error,
since it is only an estimate of it, and what remains is the
second-order error. It too is a function of $s$, so we fill it in
as well, which leaves a third-order error, and so on to the $m$th.
We show that after $m$ rounds the remaining error is at most
$1/H_m$ times the mean report $s/n$, where
$H_m=1+\frac12+\frac13+\cdots+\frac1m$ is the $m$th harmonic number
(Lemma~\ref{pp:lem:approximation}). Each round helps only a little,
but $H_m\to\infty$. The scale matters: the error of the average
estimate is multiplied by $n-1$ in the revenue, so an error of a
vanishing fraction of $s/n$ is exactly what makes the welfare loss
a vanishing fraction of $S(\theta)$.

\paragraph{Replacing the inadmissible draw.}
The mechanism draws from the other $n-1$ reports, not from all $n$.
Each draw is affected by about one part in $n$, and the $m$-round
estimate involves fewer than $m2^m$ draws, so the swap costs at most
about $m2^m/n$ mean reports (Lemma~\ref{pp:lem:deletion}). Hence a
trade-off: more rounds shrink the harmonic-number term but inflate
this cost. Taking $m\approx\frac12\log_2n$ balances them: the
total error of the mechanism's estimate is of order $1/\log\log n$
mean reports, and after the budget repair below this becomes the
rate at which $1-\alpha(M_n)$ tends to zero.

\paragraph{Budget repair.}
The sign of the final error is uncertain, so the mechanism may run
a deficit. But the error is bounded, and the bound is proportional
to the total $s$. So we charge every agent her share of the bound.
Agent $i$ does not know $s$, so her charge is computed from
$s_{-i}$, which is admissible; summed over agents these charges
make up the bound exactly, so the final error is one-sided and the
mechanism is non-deficit (Section~\ref{pp:sec:mechanism}). This is
the additive repair of~\cite{pp-guo2023lottery}: a certified
deficit bound is covered by charges independent of each agent's
own report. The magnitude
of the welfare loss is at most doubled, and a constant factor is
nothing when the magnitude goes to zero.

\section{Constructing the Shortfall Estimate}\label{pp:sec:construction}

Fix $m\geq1$, the number of rounds of error correction described
in Section~\ref{pp:sec:idea}. Agent $i$ knows the shortfall
function $G$ and her partial total $s_{-i}$; she does not know
$s$, hence not the shortfall $g=G(s)$. In this paper,
$\text{error}=\text{truth}-\text{estimate}$, not the other way
around.

\subsection{Filling in and trading}
\label{pp:sec:operators}

Let $f:\ppR\to\ppR$ be a known function. We would like to evaluate
$f$ at the total $s$; one particular case we pay attention to is
$f=G$. Each agent $i$ knows only her partial total $s_{-i}$, which
is $s$ with one report missing. Following
Section~\ref{pp:sec:idea}, the \emph{fill-in} operation replaces
the missing report by one of the $n$ reports, chosen uniformly at
random, and averages. It is a map from functions to functions,
written $\ppFill{}$, and $\ppFill{}f$ is short for $\ppFill{}$
applied to $f$:
\begin{equation}\label{pp:eq:fillin}
 (\ppFill{}f)(t)=\frac1n\sum_{j=1}^nf(t+\theta_j).
\end{equation}
Agent $i$'s fill-in estimate of $f(s)$ is $(\ppFill{}f)(s_{-i})$.
As discussed, drawing from all $n$ reports is not admissible; we
analyse it anyway and account for the difference in
Section~\ref{pp:sec:individual}.

Because the total redistribution is always the sum over all
agents, what matters is the average of the agents' estimates.
Since $s_{-i}=s-\theta_i$,
\begin{equation}\label{pp:eq:averaging}
 \begin{aligned}
 \frac1n\sum_{i=1}^n(\ppFill{}f)(s_{-i})
  &=\frac1{n^2}\sum_{i=1}^n\sum_{j=1}^nf(s-\theta_i+\theta_j)
  =(\ppSwap{}f)(s),\\
 \text{where}\quad
 (\ppSwap{}f)(t)&=\frac1{n^2}\sum_{j=1}^n\sum_{k=1}^nf(t-\theta_j+\theta_k).
 \end{aligned}
\end{equation}
The map $\ppSwap{}$ is the \emph{trade} operation: it evaluates $f$
after one report drawn from all $n$ is removed from the total and
another is added. It is the other construct our analysis uses,
and in fact the one used most; like filling in, it draws from all
$n$ reports and is not admissible as it stands, which again does
not matter for the analysis. Equation~\eqref{pp:eq:averaging} is
the link between the two: filling in, when aggregated over all
agents, each of whom has access only to $\theta_{-i}$, becomes
trading. Both maps are linear. Repeated application is written as
a power, $\ppSwap{}^2f=\ppSwap{}(\ppSwap{}f)$, and $\ppSwap{}^0=\ppId$
is the identity.

\subsection{Error, then error of the error}\label{pp:sec:error}

We follow the narrative of Section~\ref{pp:sec:idea}, now with
formulas. The agents fill in the shortfall function $G$ itself.
By~\eqref{pp:eq:averaging}, their aggregated estimate of $g=G(s)$
is $G_1(s)$ and its error is $E_1(s)$, where
\begin{equation}\label{pp:eq:round1}
 G_1=\ppSwap{}G,
 \qquad
 E_1=G-G_1=(\ppId-\ppSwap{})G .
\end{equation}
We call $\varepsilon_1=E_1(s)$ the \emph{first-order error}.
Equation~\eqref{pp:eq:round1} is the observation that drives the
construction: $E_1$ is a known function, built from $G$ and
$\ppSwap{}$ alone, and only its value at the unknown $s$ is
unknown. This is exactly the situation we started in with $G$ and
$g$.

So the agents fill in $E_1$ as well, to estimate the first-order
error, and add the estimate to their estimate of $g$.
By~\eqref{pp:eq:averaging} the aggregated estimate becomes
$G_2(s)$, and after offsetting the estimated error the residual is
the \emph{second-order error} $\varepsilon_2=E_2(s)$, where
\begin{equation}\label{pp:eq:round2}
 G_2=G_1+\ppSwap{}E_1=2\ppSwap{}G-\ppSwap{}^2G,
 \qquad
 E_2=G-G_2=(\ppId-\ppSwap{})^2G .
\end{equation}
Thus $G_2$ estimates $G$ via $G_1$ and a correction for $G_1$'s
error $E_1$. The same step once more, filling in $E_2$, gives
\begin{equation}\label{pp:eq:round3}
 G_3=G_2+\ppSwap{}E_2=3\ppSwap{}G-3\ppSwap{}^2G+\ppSwap{}^3G,
 \qquad
 E_3=G-G_3=(\ppId-\ppSwap{})^3G .
\end{equation}
The pattern in \eqref{pp:eq:round1}--\eqref{pp:eq:round3} is
that each round adds $\ppSwap{}$ of the current error function to
the estimate and multiplies the error function by
$\ppId-\ppSwap{}$.

\begin{definition}[Error functions and aggregated estimates]
\label{pp:def:correction}
For $k\geq1$, the \emph{$k$th-order error function} and the
\emph{aggregated estimate after $k$ rounds} are
\begin{equation}\label{pp:eq:aggregate-estimate}
 E_k=(\ppId-\ppSwap{})^kG,
 \qquad
 G_k=G-E_k,
\end{equation}
with values $\varepsilon_k=E_k(s)$ and $g_k=G_k(s)$ at the total,
so that $g=g_k+\varepsilon_k$.
\end{definition}

The round rule $G_{k+1}=G_k+\ppSwap{}E_k$ holds for every $k$,
because $E_k-E_{k+1}=\ppSwap{}E_k$; iterating it from
$G_1=\ppSwap{}G$ gives $G_m=\ppSwap{}(G+E_1+\cdots+E_{m-1})$. That
is, after $m$ rounds every agent has filled in
$G,E_1,\ldots,E_{m-1}$, and by~\eqref{pp:eq:averaging} the
aggregated estimate is $g_m$ with error
$\varepsilon_m$.\footnote{Estimating the error of an estimate and
adding it back, repeatedly, is the iterated bootstrap bias
correction of Hall and Martin~\cite{pp-hallmartin1988}, whose
error after $k$ rounds is likewise $(\ppId-\ppSwap{})^kG$ for a
suitable averaging operator $\ppSwap{}$. There the gain per round
is a power of the sample size, for smooth targets; here the target
has a kink at the project cost and the gain is the harmonic factor
of Lemma~\ref{pp:lem:approximation}.} By~\eqref{pp:eq:welfare}, a
mechanism using these estimates would retain revenue
$(n-1)(g_m-g)=-(n-1)\varepsilon_m$, which may have either sign.
The following lemma says that the error is at most $1/H_m$ mean
reports in absolute value.

\begin{lemma}[Error after $m$ rounds]\label{pp:lem:approximation}
For every $n\geq3$, $\theta\in[0,1]^n$, and $m\geq1$,
\begin{equation}\label{pp:eq:approximation}
 |\varepsilon_m|\leq\frac{s}{nH_m}.
\end{equation}
\end{lemma}

The proof makes use of the Poisson--Jensen inequality for bounded
analytic functions~\cite{pp-garnett2007}, a result from complex
analysis, so in this proof, and here only, we introduce constructs
involving the imaginary unit $\mathrm i$, with $\mathrm i^2=-1$;
in this proof $\mathrm i$ never denotes an agent.

\paragraph{Setup.}
Let $X,Y$ be independent draws from all $n$ reports, uniformly
and with replacement, and let $D=Y-X$, so that
$(\ppSwap{}f)(t)=\ppE\,f(t+D)$ by~\eqref{pp:eq:averaging}. The
trade $D$ is symmetric, since $X$ and $Y$ have the same law, so
$\ppE\sin(uD)=0$ for every real $u$. Define
\begin{equation}\label{pp:eq:qdef}
 q(u)=\ppE\cos(uD),\qquad u\in\ppR .
\end{equation}
Writing out the expectation over the $n^2$ equally likely pairs
$(X,Y)=(\theta_j,\theta_k)$ and using
$\cos(a-b)=\cos a\cos b+\sin a\sin b$, the double sum factors into
products of single sums:
\begin{equation}\label{pp:eq:q-modulus}
 \begin{aligned}
 q(u)&=\frac1{n^2}\sum_{j=1}^n\sum_{k=1}^n
       \cos\bigl(u\theta_k-u\theta_j\bigr)\\
     &=\frac1{n^2}\sum_{j,k}\bigl[\cos(u\theta_k)\cos(u\theta_j)
       +\sin(u\theta_k)\sin(u\theta_j)\bigr]\\
     &=\Bigl(\frac1n\sum_{j}\cos(u\theta_j)\Bigr)^2
      +\Bigl(\frac1n\sum_{j}\sin(u\theta_j)\Bigr)^2
     =\Bigl|\frac1n\sum_{j=1}^ne^{\mathrm iu\theta_j}\Bigr|^2 ,
 \end{aligned}
\end{equation}
the last step because $|C+\mathrm iS|^2=C^2+S^2$ for real $C,S$.
Three consequences of~\eqref{pp:eq:q-modulus}: $q\geq0$, being a
squared modulus; $q\leq1$, because an average of $n$ unit complex
numbers has modulus at most $1$; and $q(0)=1$, since every term
is $1$ at $u=0$.

\paragraph{Two facts.}
The first is the classical integral, valid for every real $b$,
\begin{equation}\label{pp:eq:fejer}
 \frac1\pi\int_0^\infty\frac{1-\cos(bu)}{u^2}\,\mathrm du=\frac{|b|}2 .
\end{equation}
To see it, integrate by parts with $\mathrm d(-1/u)$: the boundary
terms vanish, since $1-\cos(bu)=O(u^2)$ at $0$ and is bounded at
$\infty$, leaving $\int_0^\infty b\sin(bu)/u\,\mathrm du$, which is
$\frac\pi2|b|$ by the Dirichlet integral
$\int_0^\infty\sin(v)/v\,\mathrm dv=\frac\pi2$ after the
substitution $v=|b|u$. The second fact is the following lemma.

\begin{lemma}[Logarithmic integral bound]\label{pp:lem:log}
The integral below converges and
\begin{equation}\label{pp:eq:log}
 \frac1\pi\int_0^\infty\frac{-\log q(u)}{u^2}\,\mathrm du\leq\frac sn .
\end{equation}
\end{lemma}

The lemma is an instance of what control theory calls a Bode
integral~\cite[Chapter~3]{pp-seron1997}; we include a proof for
completeness, from the following classical inequality.

\begin{quote}
\emph{Poisson--Jensen inequality}
\cite[Chapter~II, Theorem~4.1 and Eq.~(4.3)]{pp-garnett2007}.
Let $\psi$ be analytic and bounded on the upper half plane
$\{w:\operatorname{Im}w>0\}$, continuous up to the real axis, and
not identically zero. Then for every $y>0$,
\begin{equation}\label{pp:eq:poisson-jensen}
 \log|\psi(\mathrm iy)|
 \leq\frac1\pi\int_{-\infty}^{\infty}
   \frac{y\,\log|\psi(u)|}{u^2+y^2}\,\mathrm du .
\end{equation}
\end{quote}

\begin{proof}[Proof of Lemma~\ref{pp:lem:log}]
\smallskip\noindent\emph{Step 1: a bounded analytic function.}\par\nopagebreak
Define $\psi(w)=\frac1n\sum_{j=1}^ne^{\mathrm iw\theta_j}$ for complex
$w$. It is a finite sum of exponentials, hence entire. For
$w=u+\mathrm iy$ with $y\geq0$ and every report $\theta_j\geq0$,
\[
 |e^{\mathrm iw\theta_j}|=|e^{\mathrm iu\theta_j}|\,|e^{-y\theta_j}|
 =e^{-y\theta_j}\leq1,
\]
so $|\psi(w)|\leq1$ on the closed upper half plane. It is continuous everywhere, and $\psi(0)=1$, so it is not
identically zero. Thus \eqref{pp:eq:poisson-jensen} applies to
$\psi$. Multiplying it by $-1$,
\begin{equation}\label{pp:eq:pj-applied}
 -\log|\psi(\mathrm iy)|
 \geq\frac1\pi\int_{-\infty}^{\infty}
   \frac{y\,\bigl(-\log|\psi(u)|\bigr)}{u^2+y^2}\,\mathrm du
 \qquad(y>0).
\end{equation}

\smallskip\noindent\emph{Step 2: the left-hand side.}\par\nopagebreak
At $w=\mathrm iy$ the exponentials are real:
$\psi(\mathrm iy)=\frac1n\sum_je^{-y\theta_j}$. The function
$x\mapsto e^{-yx}$ is convex, so by Jensen's inequality the average
of its values is at least its value at the average,
\[
 \psi(\mathrm iy)=\frac1n\sum_{j=1}^ne^{-y\theta_j}
 \geq e^{-y\cdot\frac1n\sum_j\theta_j}=e^{-ys/n}>0 .
\]
Hence
\[
 -\log|\psi(\mathrm iy)|=-\log\psi(\mathrm iy)\leq\frac{ys}n .
\]
Combining with~\eqref{pp:eq:pj-applied} and dividing by $y>0$,
\begin{equation}\label{pp:eq:pj-divided}
 \frac1\pi\int_{-\infty}^{\infty}
   \frac{-\log|\psi(u)|}{u^2+y^2}\,\mathrm du\leq\frac sn
 \qquad\text{for every }y>0 .
\end{equation}

\smallskip\noindent\emph{Step 3: let $y\downarrow0$.}\par\nopagebreak
The integrand in~\eqref{pp:eq:pj-divided} is nonnegative, because
$|\psi(u)|\leq1$, and for each fixed $u\ne0$ it increases to
$-\log|\psi(u)|/u^2$ as $y\downarrow0$. By the monotone convergence
theorem the integral converges to the integral of the limit, which
is therefore finite and satisfies
\begin{equation}\label{pp:eq:pj-limit}
 \frac1\pi\int_{-\infty}^{\infty}\frac{-\log|\psi(u)|}{u^2}\,\mathrm du
 \leq\frac sn .
\end{equation}
This is where convergence comes from: the limit integral is
bounded by $s/n$, so it is finite.

\smallskip\noindent\emph{Step 4: back to $q$.}\par\nopagebreak
Two facts convert~\eqref{pp:eq:pj-limit} into~\eqref{pp:eq:log}.
First, the reports are real, so $\overline{\psi(u)}=\psi(-u)$ for
real $u$, hence $|\psi(-u)|=|\psi(u)|$: the integrand in
\eqref{pp:eq:pj-limit} is an even function of $u$, and its integral
over $\ppR$ is twice its integral over $(0,\infty)$. Second, by
\eqref{pp:eq:q-modulus}, $q(u)=|\psi(u)|^2$ for real $u$, so
$-\log q(u)=-2\log|\psi(u)|$. Using the second fact and then the
first,
\[
 \begin{aligned}
 \frac1\pi\int_0^\infty\frac{-\log q(u)}{u^2}\,\mathrm du
 &=\frac1\pi\int_0^\infty\frac{-2\log|\psi(u)|}{u^2}\,\mathrm du\\
 &=\frac1\pi\int_{-\infty}^{\infty}\frac{-\log|\psi(u)|}{u^2}\,\mathrm du
 \leq\frac sn ,
 \end{aligned}
\]
which is~\eqref{pp:eq:log}.
\end{proof}

\begin{proof}[Proof of Lemma~\ref{pp:lem:approximation}]
\smallskip\noindent\emph{Step 1: an integral formula for $G$.}\par\nopagebreak
For every real $t$,
\[
 G(t)=\max\{1-t,0\}=\frac{|t-1|-(t-1)}2 .
\]
Applying~\eqref{pp:eq:fejer} with $b=t-1$ to the term $|t-1|/2$,
\begin{equation}\label{pp:eq:hinge}
 G(t)=\frac1\pi\int_0^\infty
          \frac{1-\cos((t-1)u)}{u^2}\,\mathrm du-\frac{t-1}{2}.
\end{equation}

\smallskip\noindent\emph{Step 2: apply $(\ppId-\ppSwap{})^m$ to both sides.}\par\nopagebreak
We treat both sides of~\eqref{pp:eq:hinge} as functions of $t$.
Since $\ppSwap{}$ is a finite average over shifts of $t$, it may be
applied under the integral sign. Two observations.

First, $\ppSwap{}$ fixes affine functions of $t$: for $f(t)=at+c$,
\[
 (\ppSwap{}f)(t)=\ppE\,(a(t+D)+c)=at+c+a\,\ppE D=f(t),
\]
because $\ppE D=\ppE Y-\ppE X=0$. Hence
$(\ppId-\ppSwap{})f=0$ for affine $f$, and so
$(\ppId-\ppSwap{})^mf=0$ for every $m\geq1$. In~\eqref{pp:eq:hinge}
this applies to the term $-(t-1)/2$ and to the constant $1$ in
the numerator, which is affine in $t$ as well.

Second, write $c_u(t)=\cos((t-1)u)$ for the cosine
in~\eqref{pp:eq:hinge}. By the addition formula
$\cos(\alpha+\beta)=\cos\alpha\cos\beta-\sin\alpha\sin\beta$ with
$\alpha=(t-1)u$ and $\beta=Du$, and then $\ppE\sin(uD)=0$ and
$\ppE\cos(uD)=q(u)$,
\begin{equation}\label{pp:eq:eigen}
 \begin{aligned}
 (\ppSwap{}c_u)(t)&=\ppE\cos\bigl((t-1+D)u\bigr)\\
 &=\cos((t-1)u)\,\ppE\cos(uD)-\sin((t-1)u)\,\ppE\sin(uD)\\
 &=q(u)\,c_u(t).
 \end{aligned}
\end{equation}
That is, $\ppSwap{}$ acts on $c_u$ by multiplication by the number
$q(u)$. Therefore
\[
 ((\ppId-\ppSwap{})c_u)(t)=c_u(t)-q(u)c_u(t)=(1-q(u))\,c_u(t),
\]
and applying $\ppId-\ppSwap{}$ again multiplies by $1-q(u)$ once more,
so by induction
\begin{equation}\label{pp:eq:eigen-power}
 \bigl((\ppId-\ppSwap{})^mc_u\bigr)(t)=(1-q(u))^m\,c_u(t).
\end{equation}

Now apply $(\ppId-\ppSwap{})^m$, $m\geq1$, to both sides
of~\eqref{pp:eq:hinge}, under the integral sign on the right. By
the first observation the term $-(t-1)/2$ and the constant $1$
inside the integral are annihilated; by~\eqref{pp:eq:eigen-power}
the cosine $-c_u(t)$ becomes $-(1-q(u))^mc_u(t)$. The left-hand
side becomes $((\ppId-\ppSwap{})^mG)(t)=E_m(t)$ by
Definition~\ref{pp:def:correction}. So
\[
 E_m(t)=-\frac1\pi\int_0^\infty\frac{(1-q(u))^m\,c_u(t)}{u^2}\,\mathrm du
 \qquad\text{for every real }t,
\]
and evaluating at $t=s$, where $E_m(s)=\varepsilon_m$,
\begin{equation}\label{pp:eq:residual-integral}
 \varepsilon_m=-\frac1\pi\int_0^\infty
       \frac{\cos((s-1)u)\,(1-q(u))^m}{u^2}\,\mathrm du .
\end{equation}
The integrand is integrable. For large $u$ this is because
$0\leq1-q(u)\leq1$ and $1/u^2$ is integrable at $\infty$. Near
$u=0$, $q$ is a finite sum of cosines, hence smooth, and it is an
even function of $u$ because cosine is even, so $q'(0)=0$ and
$q(u)=q(0)+O(u^2)=1+O(u^2)$; thus $1-q(u)=O(u^2)$ cancels the
$1/u^2$. This is why we passed through the integral
formula~\eqref{pp:eq:hinge}: $\ppSwap{}$ does not act on $G$ in any
simple way, but it acts on each cosine $c_u$ by a scalar factor,
so $m$ rounds of correction produce the single factor
$(1-q(u))^m$ instead of an $m$-fold alternating sum.

\smallskip\noindent\emph{Step 3: a scalar inequality.}\par\nopagebreak
For $0<x\leq1$ put $y=1-x\in[0,1)$. Since $y^k\geq y^m$ for
$1\leq k\leq m$,
\[
 H_my^m=\sum_{k=1}^m\frac{y^m}k\leq\sum_{k=1}^m\frac{y^k}k .
\]
The right-hand side is a truncation of the series
$-\log(1-y)=\sum_{k\geq1}y^k/k$, whose terms are nonnegative, so
\begin{equation}\label{pp:eq:scalar}
 H_my^m\leq\sum_{k=1}^m\frac{y^k}k\leq\sum_{k=1}^\infty\frac{y^k}k
 =-\log(1-y)=-\log x .
\end{equation}

\smallskip\noindent\emph{Step 4: conclude.}\par\nopagebreak
In~\eqref{pp:eq:residual-integral} bound $|\cos((s-1)u)|\leq1$,
then apply~\eqref{pp:eq:scalar} with $x=q(u)$ to get
$(1-q(u))^m\leq-\log q(u)/H_m$, then apply
Lemma~\ref{pp:lem:log}:
\[
 \begin{aligned}
 |\varepsilon_m|
 &\leq\frac1\pi\int_0^\infty\frac{(1-q(u))^m}{u^2}\,\mathrm du
 \leq\frac1\pi\int_0^\infty\frac{-\log q(u)}{H_m\,u^2}\,\mathrm du\\
 &=\frac1{H_m}\cdot\frac1\pi\int_0^\infty\frac{-\log q(u)}{u^2}\,\mathrm du
 \leq\frac1{H_m}\cdot\frac sn .
 \end{aligned}
\]
\end{proof}

\subsection{The mechanism's estimate: replace $\ppFill{},\ppSwap{}$ by $\ppFill{-i},\ppSwap{-i}$}
\label{pp:sec:individual}

Using $\ppFill{}$ and $\ppSwap{}$, agent $i$'s estimate of $g$ after
$m$ rounds is, by Section~\ref{pp:sec:error},
\begin{equation}\label{pp:eq:full-estimate}
 \bigl(\ppFill{}(G+E_1+\cdots+E_{m-1})\bigr)(s_{-i}).
\end{equation}
This is not admissible, because $\ppFill{}$ and $\ppSwap{}$ draw from
all $n$ reports, including $\theta_i$. Let $\ppFill{-i}$ and
$\ppSwap{-i}$ be the fill-in and trade operations on the other
$n-1$ reports,
\begin{equation}\label{pp:eq:own-trade}
 \begin{aligned}
 (\ppFill{-i}f)(t)&=\frac1{n-1}\sum_{j\ne i}f(t+\theta_j),\\
 (\ppSwap{-i}f)(t)&=\frac1{(n-1)^2}\sum_{j\ne i}\sum_{k\ne i}
   f(t-\theta_j+\theta_k).
 \end{aligned}
\end{equation}
Replacing $\ppFill{},\ppSwap{}$ in~\eqref{pp:eq:full-estimate} by
$\ppFill{-i},\ppSwap{-i}$ makes it admissible. With
$E_{k,-i}=(\ppId-\ppSwap{-i})^kG$ for $k\geq1$, the mechanism's
estimate is
\begin{equation}\label{pp:eq:own-estimate}
 \begin{aligned}
 \phi(\theta_{-i})
  &=\bigl(\ppFill{-i}(G+E_{1,-i}+\cdots+E_{m-1,-i})\bigr)(s_{-i})\\
  &=\Bigl(\ppFill{-i}\sum_{k=0}^{m-1}(\ppId-\ppSwap{-i})^kG\Bigr)(s_{-i}).
 \end{aligned}
\end{equation}
Every input to $\phi(\theta_{-i})$ is a report other than
$\theta_i$.

The replacement should cost little: a draw from all $n$ reports
picks agent $i$'s own report with probability $1/n$, so as
$n\to\infty$ it becomes impossible to draw one's own report and
the two estimates coincide. The above is admittedly hand-waving;
it serves as high-level intuition for why, as $n\to\infty$,
replacing $\ppFill{},\ppSwap{}$ by $\ppFill{-i},\ppSwap{-i}$ does
not matter much. Below we formally quantify the change it causes.

Fix an agent $i$. The quantity that will appear is the
\emph{one-draw gap}
\[
 d_i=\frac1n\Bigl(\theta_i+\frac{s_{-i}}{n-1}\Bigr),
\]
which will bound the difference between an admissible draw and an
inadmissible one. For a single agent it can be as large as about
$1/n$, but averaged over all agents, using $\sum_is_{-i}=(n-1)s$,
\begin{equation}\label{pp:eq:di-average}
 \frac1n\sum_{i=1}^nd_i
 =\frac1{n^2}\Bigl(\sum_i\theta_i+\frac1{n-1}\sum_is_{-i}\Bigr)
 =\frac1{n^2}(s+s)=\frac{2s}{n^2},
\end{equation}
and only the average enters the revenue.

\begin{lemma}[One-draw gap]\label{pp:lem:onedraw}
Let $F$ be $L$-Lipschitz. For every $t$,
$|(\ppFill{}F)(t)-(\ppFill{-i}F)(t)|\leq Ld_i$ and
$|(\ppSwap{}F)(t)-(\ppSwap{-i}F)(t)|\leq2Ld_i$.
\end{lemma}

\begin{proof}
For any numbers $v(1),\ldots,v(n)$, the average over all $n$ and
the average over the $n-1$ indices other than $i$ differ by
\begin{equation}\label{pp:eq:onedraw}
 \frac1n\sum_{j=1}^nv(j)-\frac1{n-1}\sum_{j\ne i}v(j)
 =\frac1{n(n-1)}\sum_{j\ne i}\bigl(v(i)-v(j)\bigr),
\end{equation}
as one checks by putting both sides over the denominator
$n(n-1)$. Take $v(j)=F(t+\theta_j)$, so that the left-hand side is
$(\ppFill{}F)(t)-(\ppFill{-i}F)(t)$. Since $F$ is $L$-Lipschitz,
$|v(i)-v(j)|\leq L|\theta_i-\theta_j|\leq L(\theta_i+\theta_j)$, so
the right-hand side is at most
\[
 \frac{L}{n(n-1)}\sum_{j\ne i}(\theta_i+\theta_j)
 =\frac{L}{n(n-1)}\bigl((n-1)\theta_i+s_{-i}\bigr)=Ld_i
\]
in absolute value. The trade inequality is proved similarly: a
trade is two draws, so switch the list for one draw at a time,
holding the other fixed, and apply the same bound twice.
\end{proof}

\begin{lemma}[Iterated trades]\label{pp:lem:iterated}
For every $k\geq0$, the functions $\ppSwap{}^kG$ and
$\ppSwap{-i}^kG$ are $1$-Lipschitz, and
\[
 |(\ppSwap{}^kG)(t)-(\ppSwap{-i}^kG)(t)|\leq2kd_i
 \qquad\text{for every }t .
\]
\end{lemma}

\begin{proof}
An average of shifts of a $1$-Lipschitz function is
$1$-Lipschitz, and $\ppSwap{}^kG$ and $\ppSwap{-i}^kG$ are such
averages of $G$, which gives the first claim. For the second we
induct on $k$. The claim to prove at stage $k$ is that
$\ppSwap{}^kG$ and $\ppSwap{-i}^kG$ differ by at most $2kd_i$ at
every point. For $k=0$ both are $G$ and the difference is $0$.
Assume the claim at stage $k$. Write $F=\ppSwap{}^kG$ and
$F_{-i}=\ppSwap{-i}^kG$, so that $\ppSwap{}^{k+1}G=\ppSwap{}F$ and
$\ppSwap{-i}^{k+1}G=\ppSwap{-i}F_{-i}$. Insert the mixed term
$\ppSwap{}F_{-i}$:
\[
 |\ppSwap{}F-\ppSwap{-i}F_{-i}|
 \leq|\ppSwap{}F-\ppSwap{}F_{-i}|+|\ppSwap{}F_{-i}-\ppSwap{-i}F_{-i}|
\]
pointwise. The first difference is $\ppSwap{}(F-F_{-i})$, an
average of values each at most $2kd_i$ by the stage-$k$ claim, so
it is at most $2kd_i$. The second compares $\ppSwap{}$ and
$\ppSwap{-i}$ on the same function $F_{-i}$, which is
$1$-Lipschitz by the first claim, so it is at most $2d_i$ by
Lemma~\ref{pp:lem:onedraw} with $L=1$. Adding gives $2(k+1)d_i$,
the claim at stage $k+1$.
\end{proof}

\begin{lemma}[Replacement cost]
\label{pp:lem:deletion}
For every $n\geq3$, $\theta\in[0,1]^n$, and $m\geq1$,
\begin{equation}\label{pp:eq:deletion}
 |\Phi(\theta)-g_m|\leq\frac{4m2^m\,s}{n^2},
 \qquad\text{where}\quad\Phi(\theta)=\frac1n\sum_i\phi(\theta_{-i}).
\end{equation}
\end{lemma}

\begin{proof}
Let $F=G+E_1+\cdots+E_{m-1}$ and
$F_{-i}=G+E_{1,-i}+\cdots+E_{m-1,-i}$, so that
$\phi(\theta_{-i})=(\ppFill{-i}F_{-i})(s_{-i})$
by~\eqref{pp:eq:own-estimate}, while by
Definition~\ref{pp:def:correction} and~\eqref{pp:eq:averaging},
$g_m=G_m(s)=(\ppSwap{}F)(s)=\frac1n\sum_i(\ppFill{}F)(s_{-i})$.
Therefore
\[
 |\Phi(\theta)-g_m|
 \leq\frac1n\sum_{i=1}^n
   \bigl|(\ppFill{-i}F_{-i})(s_{-i})-(\ppFill{}F)(s_{-i})\bigr| .
\]
We bound the $i$th term by $2m2^md_i$; averaging over $i$ with
\eqref{pp:eq:di-average} then gives the lemma.

Multiplying out the powers in $E_k=(\ppId-\ppSwap{})^kG$,
$F=\sum_{k=0}^{m-1}c_k\ppSwap{}^kG$ for integer coefficients
$c_k$, and $F_{-i}=\sum_{k=0}^{m-1}c_k\ppSwap{-i}^kG$ with the same
coefficients. The absolute values of the coefficients of
$(\ppId-\ppSwap{})^k$ sum to $2^k$, so
$\sum_k|c_k|\leq\sum_{k<m}2^k<2^m$. Two consequences: by
Lemma~\ref{pp:lem:iterated}, at every point
\[
 |F-F_{-i}|\leq\sum_k|c_k|\,2kd_i\leq2^m\cdot2(m-1)d_i ,
\]
and $F_{-i}$, a combination of $1$-Lipschitz functions with
coefficients of total absolute value below $2^m$, is
$2^m$-Lipschitz. Now insert the mixed term $\ppFill{}F_{-i}$:
\[
 \begin{aligned}
 |\ppFill{}F-\ppFill{-i}F_{-i}|
 &\leq|\ppFill{}F-\ppFill{}F_{-i}|+|\ppFill{}F_{-i}-\ppFill{-i}F_{-i}|\\
 &\leq2^m\cdot2(m-1)d_i+2^md_i\leq2m2^md_i .
 \end{aligned}
\]
The first term is $\ppFill{}(F-F_{-i})$, an average of values
bounded by the display above. The second is
Lemma~\ref{pp:lem:onedraw} with $L=2^m$, applied to the
$2^m$-Lipschitz function $F_{-i}$.
\end{proof}

Since $g_m=g-\varepsilon_m$, Lemmas~\ref{pp:lem:approximation}
and~\ref{pp:lem:deletion} give, at every profile,
\begin{equation}\label{pp:eq:aggregate}
 |\Phi(\theta)-g|
 \leq|\varepsilon_m|+|\Phi(\theta)-g_m|
 \leq b_{n,m}\,\frac sn,
 \qquad
 b_{n,m}=\frac1{H_m}+\frac{4m2^m}n .
\end{equation}
We call $b_{n,m}$ the \emph{error allowance}; it depends only on
$n$ and $m$ and is therefore known before the reports arrive.

\section{The Mechanism and Its Guarantee}\label{pp:sec:mechanism}

\begin{definition}[Mechanism $M_{n,m}$]\label{pp:def:mechanism}
Fix integers $n\geq3$ and $m\geq1$ before the reports arrive.
Build the project if and only if $s\geq1$, and use the
payments~\eqref{pp:eq:payments} with
\begin{equation}\label{pp:eq:mechanism}
 h(\theta_{-i})
   =s_{-i}+\frac{n-1}{n}\,\phi(\theta_{-i})+\frac{b_{n,m}}ns_{-i}.
\end{equation}
\end{definition}

The first term is the linear part of the ideal Groves term of
Section~\ref{pp:sec:idea}, the second is the mechanism's estimate
of its shortfall term $g$, and the third covers the possible
deficit, as follows. Without the third term, the sum of the Groves
terms over agents would be $(n-1)s+(n-1)\Phi(\theta)$, and
by~\eqref{pp:eq:welfare} the retained revenue would be
$(n-1)[\Phi(\theta)-g]$, which can be negative; by
\eqref{pp:eq:aggregate} the deficit is at most $(n-1)b_{n,m}s/n$.
The third term covers exactly this amount: it depends only on
$\theta_{-i}$, so it is admissible, and summed over agents,
\[
 \sum_i\frac{b_{n,m}}ns_{-i}=\frac{n-1}nb_{n,m}s .
\]

\begin{theorem}[No deficit and bounded welfare loss]\label{pp:thm:main}
For every $n\geq3$ and $m\geq1$, the mechanism $M_{n,m}$ is
deterministic, anonymous, efficient, dominant-strategy incentive
compatible, and ex post weakly budget balanced. At every
profile $\theta\in[0,1]^n$,
\begin{equation}\label{pp:eq:mainbound}
 0\leq L(\theta)\leq2\,b_{n,m}\,S(\theta),
\end{equation}
and consequently $\alpha(M_{n,m})\geq1-2b_{n,m}$.
\end{theorem}

\begin{proof}
The decision is efficient and the payments have Groves form, so
the argument of Section~\ref{pp:sec:model} gives
dominant-strategy incentive compatibility; the other listed
properties are immediate from the definition.
Summing~\eqref{pp:eq:mechanism} over agents and
using~\eqref{pp:eq:welfare},
\begin{equation}\label{pp:eq:budget}
 L(\theta)=(n-1)[\Phi(\theta)-g]+\frac{n-1}nb_{n,m}s .
\end{equation}
By~\eqref{pp:eq:aggregate} the first term is at most the second
in absolute value, so
$0\leq L(\theta)\leq\frac{2(n-1)}nb_{n,m}s\leq2b_{n,m}S(\theta)$.
The left inequality is weak budget balance, and
$W(\theta)/S(\theta)=1-L(\theta)/S(\theta)$ gives the ratio.
\end{proof}

\subsection{Choosing the number of rounds}

The allowance $b_{n,m}$ in~\eqref{pp:eq:aggregate} exhibits a
trade-off. More rounds shrink the harmonic-number term $1/H_m$
but inflate the replacement cost $4m2^m/n$, which grows
exponentially in $m$. A larger population makes the replacement
cost smaller and so permits more rounds. We take
\[
 m_n=\max\{1,\lfloor\tfrac12\log_2n\rfloor\},
 \qquad M_n=M_{n,m_n},
\]
so that $2^{m_n}\leq\sqrt n$ for $n\geq4$ and both terms vanish
as $n\to\infty$.

\begin{proof}[Proof of Theorem~\ref{pp:thm:asymptotic}]
For $n\geq4$, $m_n\leq\frac12\log_2n$ and $2^{m_n}\leq\sqrt n$, so
the replacement cost satisfies
$4m_n2^{m_n}/n\leq4m_n/\sqrt n=O(\log n/\sqrt n)$.
Integral comparison gives $\log(m+1)\leq H_m\leq1+\log m$, and
$m_n$ grows proportionally to $\log n$, so
$1/H_{m_n}=\Theta(1/\log\log n)$, which dominates. Hence
\[
 b_{n,m_n}=O\!\left(\frac1{\log\log n}\right)\longrightarrow0,
\]
and Theorem~\ref{pp:thm:main} gives
$1-\alpha(M_n)\leq2b_{n,m_n}=O(1/\log\log n)$. In particular
$\alpha(M_n)\to1$ as $n\to\infty$.
\end{proof}

\section{Discussion}\label{pp:sec:discussion}

\begin{remark}[Why bounded precision is not needed]\label{pp:rem:precision}
Section~\ref{pp:sec:mainresult} explained what the bounded
precision assumption buys the mechanism of~\cite{pp-guo2019}: on
the profiles that matter, almost every agent reports zero and so
knows the total exactly. The present mechanism never uses such a
fact. Its estimate is built from averages of the shortfall function
$G$ over the reports, every bound in
Lemmas~\ref{pp:lem:approximation} and~\ref{pp:lem:deletion} is a
multiple of the mean report $s/n$, and the surcharge is a multiple
of $s_{-i}$. Shrinking all reports by a common factor shrinks every
error bound and the surcharge by the same factor, so a profile of many
tiny positive reports, which breaks the earlier mechanism, is no
harder than a profile of a few large ones. No arithmetic property
of the reports enters at any point.
\end{remark}

\begin{remark}[Where the rate comes from]\label{pp:rem:rate}
The allowance $b_{n,m}=1/H_m+4m2^m/n$ has two terms, one from each
of the two lemmas. The first is the residual error of the
$m$-round analysis, Lemma~\ref{pp:lem:approximation}; it is only
logarithmic in $m$. The second is the price of replacing the
full-list operations by the admissible ones,
Lemma~\ref{pp:lem:deletion}; the factor $2^m$ comes from bounding
each of the $m$ error functions separately and crudely. The two
terms balance at $m$ of order $\log n$, which gives
$1/\log\log n$. Improving the replacement cost from $m2^m$ to a
polynomial in $m$ would allow $m$ of order a power of $n$ and
improve the rate to $1/\log n$; a faster-than-logarithmic decay of
the residual error would improve it further. We do not know
whether either bound is tight for this construction.
\end{remark}

\begin{remark}[Evaluation cost, and a randomised version]\label{pp:rem:cost}
Expanding $\phi(\theta_{-i})$ into averages of $G$, the term with
$k$ trades ranges over $(n-1)^{2k+1}$ ordered tuples per agent.
Enumerating them while reusing partial sums takes
$O(nm(n-1)^{2m-1})$ arithmetic operations for all agents, which is
polynomial in $n$ for fixed $m$ and $n^{O(\log n)}$ for $m=m_n$.
The construction therefore does not provide a deterministic
polynomial-time evaluation algorithm. It does provide a randomised
one: $\phi(\theta_{-i})$ is a signed combination of averages of
$G$ over tuples drawn uniformly from $\theta_{-i}$, so drawing a
polynomial number of tuples and averaging gives an unbiased
estimate of $\phi(\theta_{-i})$. The resulting payments depend only
on $\theta_{-i}$ and on the mechanism's own randomness, so
truthfulness in dominant strategies holds for every realisation,
while the non-deficit and welfare guarantees of
Theorem~\ref{pp:thm:main} hold in expectation over the mechanism's
randomness rather than ex post.
\end{remark}

\begin{remark}[Beyond the public project problem]\label{pp:rem:beyond}
Of the shortfall function $G$, the replacement-cost lemmas use only
that it is $1$-Lipschitz, and the harmonic bound uses only its
hinge shape. A hinge at any point $c$ has the same trigonometric
representation as~\eqref{pp:eq:hinge} with $t-1$ replaced by $t-c$,
and the trade operation treats it in the same way, so the whole
argument applies verbatim to any piecewise-linear function of the
total, with the constants multiplied by the total variation of its
slope. Whether the scheme helps in other redistribution problems, where
the ideal Groves term is a function of a total the agent cannot
see, we have not pursued.
\end{remark}

\section{Conclusion}\label{pp:sec:conclusion}

We have given a deterministic, anonymous, strategy-proof, and
non-deficit mechanism for the public project problem whose
worst-case efficiency ratio tends to $1$ as the number of agents
grows, for arbitrary real types in $[0,1]$. This resolves the
asymptotic question raised by Naroditskiy et
al.~\cite{pp-naroditskiy2012}, without the bounded precision
assumption under which it was previously
answered~\cite{pp-guo2019}. The mechanism itself is a finite
averaging scheme: fill in the missing report, estimate the error,
estimate the error of the error, and repeat; the analysis reduces
to one bound on an iterated averaging operator, which we prove by
representing the shortfall function as a trigonometric integral
and applying the Poisson--Jensen inequality for bounded analytic
functions.

With the asymptotic question settled, the open problem is the one
for specific numbers of agents. Exact worst-case optimal mechanisms
are known only for $n\leq5$~\cite{pp-naroditskiy2012,pp-guo2023lottery},
where they match the conjectured upper bounds
of~\cite{pp-naroditskiy2012}; for larger $n$ it is not known
whether those bounds are tight, and the present mechanism, whose
guarantee is informative only for populations in the thousands,
says nothing about them. The rate $1-O(1/\log\log n)$ is what this
construction delivers, and we make no claim that it is optimal.
Whether the mechanism's estimate, which direct enumeration
evaluates in $n^{O(\log n)}$ operations, admits a deterministic
polynomial-time computation is also open; a randomised one is
described in Section~\ref{pp:sec:discussion}.

\paragraph{Disclosure of AI assistance.}
AI assistants contributed to the methodology. The author posed the
problem and the target theorem. Claude Code (Claude Fable~5.1) did
not find a proof; OpenAI Codex (GPT-6 Astra) found the
iterated-correction construction and a first proof, which was not
human-readable. An earlier successful proof of
Lemma~\ref{pp:lem:approximation} proceeded by analysing grid-valued
reports and extending the result to arbitrary real reports by
continuity. The argument here
is a rewrite under the author's direction over more than twenty
passes: the operator presentation, the aggregate analysis with the
full list of reports and a Lipschitz correction at the end, the
division into lemmas, and the budget repair are the author's design;
the proof of Lemma~\ref{pp:lem:approximation} through the
Poisson--Jensen inequality is essentially the assistants', selected
from several candidates and rewritten by the author. The author
checked every line; no numerical mechanism search was used at any
stage. The work spanned many sessions, so no single prompt exists.

\end{document}